\documentclass[12pt]{article}

\usepackage{amsmath,amssymb,bbm,epsf,epsfig,graphicx,fullpage}
\usepackage{tikz}
\usetikzlibrary{shapes,decorations.pathreplacing,patterns}

\title{On the Containment Problem for Linear Sets}
\author{Hans U.~Simon \\ 
Department of Computer Science, Ruhr~University Bochum, Germany \\
E-Mail Address: {\tt hans.simon@rub.de}}

\begin{document}
\newtheorem{theorem}{Theorem}[section]
\newtheorem{lemma}[theorem]{Lemma}
\newtheorem{remark}[theorem]{Remark}
\newtheorem{bemerkung}[theorem]{Bemerkung}
\newtheorem{satz}[theorem]{Satz}
\newtheorem{fact}[theorem]{Fact}
\newtheorem{claim}[theorem]{Claim}
\newtheorem{definition}[theorem]{Definition}
\newtheorem{corollary}[theorem]{Corollary}
\newtheorem{folgerung}[theorem]{Folgerung}
\newtheorem{folgerungdefinition}[theorem]{Folgerung und Definition}
\newtheorem{examples}[theorem]{Examples}
\newtheorem{example}[theorem]{Example}
\newtheorem{beispiel}[theorem]{Beispiel}
\newtheorem{beobachtung}[theorem]{Beobachtung}
\newenvironment{proof}{\medskip\noindent {\bf Proof}}{\hfill$\bullet$\par\bigskip}
\newenvironment{beweis}{\medskip\noindent {\bf Beweis}}{\mbox{}\hfill$\bullet$\par\bigskip}
\newenvironment{open}{\medskip\noindent {\bf Open Problems}\begin{enumerate}}{\end{enumerate}\par\bigskip}
\newenvironment{sketch}{\medskip\noindent {\bf Sketch of Proof}}{\hfill$\bullet$\par\bigskip}
\newenvironment{beweisskizze}{\medskip\noindent {\bf Beweisskizze}}{\hfill$\Box$\par\bigskip}
\newenvironment{claimproof}{\medskip\noindent {\bf Proof of the Claim}}{\hfill$\bullet$\par\bigskip}
\newcommand{\nl}{\mbox{}\newline}
\newcommand{\MDNF}{\mbox{MDNF}}
\newcommand{\outdeg}{\mbox{deg}}
\newcommand{\indeg}{\mbox{deg}}
\newcommand{\BOX}{\mbox{BOX}}
\newcommand{\EQ}{\mbox{EQ}}
\newcommand{\MQ}{\mbox{MQ}}
\newcommand{\ACC}{\mbox{ACC}}
\newcommand{\REJ}{\mbox{REJ}}
\newcommand{\RSA}{\mbox{RSA}}
\newcommand{\INSEC}{\mbox{INSEC}}
\newcommand{\Log}{{\mathcal L}}
\newcommand{\NL}{{\mathcal NL}}
\newcommand{\POL}{\mathit{POL}}
\newcommand{\pad}{\mathit{pad}}
\newcommand{\desc}{\mathit{desc}}
\newcommand{\size}{\mathit{size}}
\newcommand{\Pol}{\mathit{P}}
\newcommand{\ppoly}{\mathit{P/poly}}
\newcommand{\NP}{\mathit{NP}}
\newcommand{\PP}{\mathit{PP}}
\newcommand{\RP}{\mathit{RP}}
\newcommand{\BPP}{\mathit{BPP}}
\newcommand{\ZPP}{\mathit{ZPP}}
\newcommand{\NPC}{\mathit{NPC}}
\newcommand{\NPI}{\mathit{NPI}}
\newcommand{\PSpace}{\mathit{PSpace}}
\newcommand{\dens}{\mathit{dens}}
\newcommand{\PH}{\mathit{PH}}
\newcommand{\IP}{{\mathcal IP}}
\newcommand{\cE}{{\mathcal E}}
\newcommand{\cNP}{{\mathcal NP}}
\newcommand{\cA}{{\mathcal A}}
\newcommand{\cB}{{\mathcal B}}
\newcommand{\cO}{{\mathcal O}}
\newcommand{\cZ}{{\mathcal Z}}
\newcommand{\cT}{{\mathcal T}}
\newcommand{\cK}{{\mathcal K}}
\newcommand{\cX}{{\mathcal X}}
\newcommand{\cY}{{\mathcal Y}}
\newcommand{\cD}{{\mathcal D}}
\newcommand{\IMQ}{\mbox{IMQ}}
\newcommand{\DSpace}{\mathit{DSpace}}
\newcommand{\NSpace}{\mathit{NSpace}}
\newcommand{\DTime}{\mathit{DTime}}
\newcommand{\NTime}{\mathit{NTime}}
\newcommand{\DTimeSpace}{\mathit{DTimeSpace}}
\newcommand{\NTimeSpace}{\mathit{NTimeSpace}}
\newcommand{\bin}{\mbox{bin}}
\newcommand{\sign}{\mbox{sign}}
\newcommand{\pred}{\mbox{pred}}
\newcommand{\partner}{\mbox{-partner}}
\newcommand{\GRP}{\mbox{GRP}}
\newcommand{\SAT}{\mbox{SAT}}
\newcommand{\KP}{\mbox{KP}}
\newcommand{\UNSAT}{\mbox{UNSAT}}
\newcommand{\MEC}{\mbox{MEC}}
\newcommand{\ziffern}{\mbox{ziffern}}
\newcommand{\rank}{\mbox{rank}}
\newcommand{\argmax}{\mathrm{argmax}}
\newcommand{\sensitivity}{\mathrm{sensitivity}}
\newcommand{\precision}{\mathrm{precision}}
\newcommand{\specificity}{\mathrm{specificity}}
\newcommand{\Gain}{\mathrm{Gain}}
\newcommand{\st}{\mathrm{s.t.}}
\newcommand{\diag}{\mathrm{diag}}
\newcommand{\TF}{\mathrm{TF}}
\newcommand{\DKW}{\mathrm{DKW}}
\newcommand{\DF}{\mathrm{DF}}
\newcommand{\LIN}{\mathrm{LIN}}
\newcommand{\WL}{\mathrm{WL}}
\newcommand{\DS}{\mathrm{DS}}
\newcommand{\SIG}{\mathrm{SIG}}
\newcommand{\HS}{\mathrm{HS}}
\newcommand{\MinDisE}{\mathrm{MinDisE}}
\newcommand{\id}{\mathrm{id}}
\newcommand{\argmin}{\mathrm{argmin}}
\newcommand{\ggT}{\mathrm{ggT}}
\newcommand{\einsDL}{1-\mathrm{DL}}
\newcommand{\dt}{\mathrm{-DT}}
\newcommand{\ID}{1-\mathrm{ID}}
\newcommand{\einsCNF}{1-\mathrm{CNF}}
\newcommand{\kDL}{k\mathrm{-DL}}
\newcommand{\kCNF}{k\mathrm{-CNF}}
\newcommand{\kclauseCNF}{k\mathrm{-clause-CNF}}
\newcommand{\kDNF}{k\mathrm{-DNF}}
\newcommand{\ktermDNF}{k\mathrm{-term-DNF}}
\newcommand{\found}{\mbox{\tt found}}
\newcommand{\fml}{\mbox{fml}}
\newcommand{\val}{\mbox{val}}
\newcommand{\TRUE}{{\mbox{TRUE}}}
\newcommand{\cex}{{\mbox{cex}}}
\newcommand{\BVar}{{\mbox{BVar}}}
\newcommand{\BFml}{{\mbox{BFml}}}
\newcommand{\F}{{\mbox{F}}}
\newcommand{\T}{{\mbox{T}}}
\newcommand{\Tm}{{\mbox{Tm}}}
\newcommand{\Fml}{{\mbox{Fml}}}
\newcommand{\HK}{{\mbox{HK}}}
\newcommand{\QR}{{\mbox{QR}}}
\newcommand{\SQUARE}{{\mbox{SQUARE}}}
\newcommand{\Aus}{{\mbox{Aus}}}
\newcommand{\Fr}{{\mbox{Fr}}}
\newcommand{\FALSE}{{\mbox{FALSE}}}
\newcommand{\GE}{\mbox{GE}}
\newcommand{\LE}{\mbox{LE}}
\newcommand{\sm}{\setminus}
\newcommand{\bom}{\bar{\omega}}
\newcommand{\seq}{\subseteq}
\newcommand{\supeq}{\supseteq}
\newcommand{\sqseq}{\sqsubseteq}
\newcommand{\sqsupeq}{\sqsupseteq}
\newcommand{\lepol}{\le_{pol}}
\newcommand{\llepol}{\le_{pol}^{L}}
\newcommand{\lelog}{\le_{log}}
\newcommand{\lepolcook}{\le_T}
\newcommand{\lepolkarp}{\le_{pol}}
\newcommand{\lesspolkarp}{<_{pol}}
\newcommand{\lepollevin}{\le_L}
\newcommand{\expol}{(\exists)_{pol}}
\newcommand{\allpol}{(\forall)_{pol}}
\newcommand{\Qpol}{(Q_k)_{pol}}
\newcommand{\negQpol}{(\bar Q_k)_{pol}}
\newcommand{\exallpol}{\expol\allpol\expol\cdots\Qpol}
\newcommand{\allexpol}{\allpol\expol\allpol\cdots\negQpol}
\newcommand{\OR}{\vee}
\newcommand{\AND}{\wedge}
\newcommand{\lra}{\leftrightarrow}
\newcommand{\Lra}{\Longrightarrow}
\newcommand{\ra}{\rightarrow}
\newcommand{\lora}{\longrightarrow}
\newcommand{\la}{\leftarrow}
\newcommand{\La}{\Leftarrow}
\newcommand{\Ra}{\Rightarrow}
\newcommand{\error}{\mathit{error}}
\newcommand{\err}{\mathit{err}}
\newcommand{\Err}{\mathit{Err}}
\newcommand{\DT}{\mbox{--DT}}
\newcommand{\LOW}{\mbox{LOW}}
\newcommand{\cost}{\mbox{cost}}
\newcommand{\pot}{\mbox{pot}}
\newcommand{\poly}{\mbox{poly}}
\newcommand{\pol}{\mbox{pol}}
\newcommand{\HIGH}{\mbox{HIGH}}
\newcommand{\OPT}{\mbox{OPT}}
\newcommand{\MLP}{\mbox{--MLP}}
\newcommand{\Om}{\Omega}
\newcommand{\om}{\omega}
\newcommand{\dund}{\Leftrightarrow}
\newcommand{\beq}{\begin{equation}}
\newcommand{\eeq}{\end{equation}}
\newcommand{\eset}{\emptyset}
\newcommand{\ol}{\overline}
\newcommand{\bool}{\{0,1\}}
\newcommand{\eps}{\epsilon}
\newcommand{\ve}{\varepsilon}
\newcommand{\true}{\mbox{true}}
\newcommand{\nil}{\mbox{\bf nil}}
\newcommand{\impl}{\Rightarrow}
\newcommand{\Exp}{\mathbbm{E}}
\newcommand{\Ind}{\mathbbm{1}}
\newcommand{\nats}{\mathbbm{N}}
\newcommand{\ints}{\mathbbm{Z}}
\newcommand{\cN}{{\mathcal N}}
\newcommand{\cG}{{\mathcal G}}
\newcommand{\cC}{{\mathcal C}}
\newcommand{\cH}{{\mathcal H}}
\newcommand{\cV}{{\mathcal V}}
\newcommand{\cW}{{\mathcal W}}
\newcommand{\cF}{{\mathcal F}}
\newcommand{\cM}{{\mathcal M}}
\newcommand{\cS}{{\mathcal S}}
\newcommand{\cI}{{\mathcal I}}
\newcommand{\cQ}{{\mathcal Q}}
\newcommand{\cR}{{\mathcal R}}
\newcommand{\TM}{\subseteq}
\newcommand{\false}{\mbox{false}}
\newcommand{\good}{\mbox{good}}
\newcommand{\bad}{\mbox{bad}}
\newcommand{\bit}{\mbox{bit}}
\newcommand{\lsb}{\mbox{lsb}}
\newcommand{\falls}{\mbox{ if}}
\newcommand{\und}{\mbox{and}}
\newcommand{\oder}{\mbox{oder}}
\newcommand{\pr}{\mbox{pr}}
\newcommand{\opt}{\mbox{opt}}
\newcommand{\Chow}{\mbox{Chow}}
\newcommand{\VCdim}{\mbox{VCdim}}
\newcommand{\VCD}{\mathrm{VCdim}}
\newcommand{\cdim}{\mbox{Cdim}}
\newcommand{\scdim}{\mbox{SCdim}}
\newcommand{\hdim}{\mbox{Hdim}}
\newcommand{\shdim}{\mbox{SHdim}}
\newcommand{\HC}{\mbox{HC}}
\newcommand{\RLC}{\mbox{RLC}}
\newcommand{\DLC}{\mbox{DLC}}
\newcommand{\BLE}{\mbox{BLE}}
\newcommand{\ADV}{\mbox{ADV}}
\newcommand{\QBF}{\mbox{QBF}}
\newcommand{\HEX}{\mbox{HEX}}
\newcommand{\qbf}{\mbox{qbf}}
\newcommand{\otherwise}{\mbox{otherwise}}
\newcommand{\supp}{\mbox{supp}}
\newcommand{\out}{\mbox{out}}
\newcommand{\er}[1]{\stackrel{#1}{\equiv}}
\newcommand{\transition}[1]{\stackrel{#1}{\longrightarrow}}
\newcommand{\righttransit}[1]{\stackrel{\longrightarrow}{#1}}
\newcommand{\lefttransit}[1]{\stackrel{\longleftarrow}{#1}}
\newcommand{\fmove}[2]{\stackrel{#1,#2}{\longrightarrow}}
\newcommand{\sprod}[2]{\langle #1 , #2 \rangle}
\newcommand{\Span}[1]{\langle #1 \rangle}
\newcommand{\cfmove}[1]{\stackrel{#1}{\longrightarrow}}
\newcommand{\abl}[1]{\stackrel{#1}{\Longrightarrow}}
\newcommand{\bmove}[1]{\stackrel{#1}{\longleftarrow}}
\newcommand{\cbmove}[2]{\stackrel{#1,#2}{\longleftarrow}}
\newcommand{\round}[3]{\fbox{$\begin{array}{c}
                         \fmove{#1}{#2} \\
                         \bmove{#3}
                       \end{array}$}}
\newcommand{\cfround}[2]{\fbox{$\begin{array}{c}
                         \cfmove{#1} \\
                         \bmove{#2}
                       \end{array}$}}
\newcommand{\poleq}{\stackrel{pol}{\sim}}
\newcommand{\poliso}{\stackrel{pol}{\simeq}}
\newcommand\reals{\mathbbm{R}}
\newcommand\rationals{\mathbbm{Q}}
\newcommand{\defeq}{:=}
\newcommand{\fac}[1]{{#1}!}
\newcommand{\co}[1]{\mbox{co-}{#1}}
\newcommand{\vexpol}[1]{(\exists{#1})_{pol}}
\newcommand{\vallpol}[1]{(\forall{#1})_{pol}}
\newcommand{\vQpol}[1]{(Q_k{#1})_{pol}}
\newcommand{\vQpolpol}[1]{(Q_{p(|x|)}{#1})_{pol}}
\newcommand{\vQpolpolq}[1]{(Q_{q(|x|)}{#1})_{pol}}
\newcommand{\vnegQpol}[1]{(\bar Q{#1})_{pol}}
\newcommand{\vexallpol}{\vexpol{y_1}\vallpol{y_2}\vexpol{y_3}\cdots\vQpol{y_k}}
\newcommand{\oddvexallpol}{\vexpol{y_1}\vallpol{y_2}\cdots\vexpol{y_k}}
\newcommand{\vexallpolp}{\vexpol{y'_1}\vallpol{y'_2}\vexpol{y'_3}\cdots\vQpol{y'_k}}
\newcommand{\vexallpoldp}{\vexpol{y''_1}\vallpol{y''_2}\vexpol{y''_3}\cdots\vQpol{y''_k}}
\newcommand{\vallexpol}{\vallpol{y_1}\vexpol{y_2}\vallpol{y_3}\cdots\vnegQpol{y_k}}
\newcommand{\vexallpolpol}{\vexpol{y_1}\vallpol{y_2}\vexpol{y_3}\cdots\vQpolpol{y_{p(|x|)}}}
\newcommand{\vexallpolpolq}{\vexpol{y_1}\vallpol{y_2}\vexpol{y_3}\cdots\vQpolpolq{y_{q(|x|)}}}

\newcommand{\sequiv}{\stackrel{s}{\equiv}}

\renewcommand{\vec}[1]{\mathbf{#1}}

\maketitle

\begin{abstract}
It is well known that the containment problem (as well as the
equivalence problem) for semilinear sets is $\log$-complete
in $\Pi_2^p$. It had been shown quite recently that already 
the containment problem for multi-dimensional linear sets is 
$\log$-complete in $\Pi_2^p$ (where hardness even holds for 
a unary encoding of the numerical input parameters). In this 
paper, we show that already the containment problem for 
$1$-dimensional linear sets (with binary encoding of the 
numerical input parameters) is $\log$-hard (and therefore 
also $\log$-complete) in $\Pi_2^p$. However, combining both 
restrictions (dimension $1$ and unary encoding), the problem 
becomes solvable in polynomial time.
\end{abstract}

\section{Introduction}

The containment problem for a family of sets consists
in finding an answer to the following question: given 
two sets of the family, is the first one a subset of 
the second one?

It had been shown in a very early stage of complexity theory
that the containment and the equivalence problem for semilinear 
sets are $\log$-complete in $\Pi_2^p$ (the second level of the
polynomial hierarchy)~\cite{H1982}. This early investigation
had been motivated by the fact that, first, the equivalence
problem for contextfree languages is recursively undecidable
and, second, the commutative images of contextfree languages
happen to be semilinear sets according to Parikh's 
theorem~\cite{P1966}. Showing inequivalence of the commutative 
images of two given contextfree languages would therefore
demonstrate their inequivalence.

Linear sets are the basic building blocks of semilinear sets.
(The latter are finite unions of linear sets.)
Moreover, $1$-dimensional linear sets are the central object 
of research in the study of numerical semigroups~\cite{RG-S2009}. 
It was shown quite recently that the containment problem 
for linear sets of variable dimension is $\log$-complete 
in $\Pi_2^p$, where hardness even holds when numbers
are encoded in unary~\cite{CHH2016}. In this paper, 
we extend the latter result as follows:
\begin{enumerate}
\item
The containment problem for $1$-dimensional linear sets
(with a binary encoding of numbers) is $\log$-hard (and 
therefore also $\log$-complete) in $\Pi_2^p$.
\item
On the other hand, the containment problem for 
$1$-dimensional linear sets with a unary encoding 
of numbers is solvable in polynomial time.
\end{enumerate}
Moreover, in order to prove these results, we show 
the following:
\begin{itemize}
\item
The containment problem for so-called simple 
unary $(+,\cup)$-expressions\footnote{a variant of a 
problem that has originally been analyzed by 
Stockmeyer~\cite{S1977}} is $\log$-hard in~$\Pi_2^p$.
\item
The containment problem for linear sets is still
$\log$-hard in $\Pi_2^p$ under a relatively strong 
promise.  See Sections~\ref{subsec:promise} 
and~\ref{sec:main-results} for details.
\end{itemize}
These results might be of independent interest.

As for semilinear sets, the containment and the inequivalence
problem have the same inherent complexity: both are
$\log$-complete in $\Pi_2^p$. We briefly note that the
situation is different for linear sets. The equivalence problem
for linear sets is easily shown to be computationally equivalent
to the word problem for linear sets, and the latter is easily
shown to be NP-complete. Hence, for linear sets, verifying
containment is much harder than verifying equivalence.

This paper is structured as follows. In Section~\ref{sec:facts}
we present the basic definitions and notations, and we mention
some facts. Our main results are stated and proved in 
Section~\ref{sec:main-results}.  One of these proofs is however 
postponed to the final Section~\ref{sec:reduction2intexp} because 
it is a suitable modification of a similar proof of Stockmeyer 
(and is given for the sake of completeness). In the final
Section~\ref{sec:open-problems}, an open problem is mentioned.

\section{Definitions, Notations and Facts} \label{sec:facts}

We assume familiarity with basic concepts from complexity theory
(e.g., logspace reductions, $\log$-hardness or $\log$-completeness,
polynomial hierarchy etc.). The complexity classes of the polynomial
hierarchy will be denoted, as usual, by $\Sigma_k^p$ and $\Pi_k^p$
for $k = 0,1,2,\ldots$. We will mainly deal with the class $\Pi_2^p$
on the second level of the hierarchy. 

In Section~\ref{subsec:qbf}, we briefly call into mind the definition
of true quantified Boolean formulas which give rise to a hierarchy of 
problems with one $\log$-complete problem at every level of the polynomial 
hierarchy. Section~\ref{subsec:integer-expressions} contains the
basic definitions that we need in connection with integer expressions.
In Section~\ref{subsec:linear-sets}, we briefly remind the reader to
the definition of linear and semilinear sets. Some well known results
on the inherent complexity of the containment problem for integer 
expressions resp.~for semilinear sets are mentioned in
Section~\ref{subsec:containment-problems}. Section~\ref{subsec:promise}
briefly calls into mind the notion of promise problems.

\subsection{Quantified Boolean Formulas} \label{subsec:qbf}

\begin{definition}[\cite{S1977}]
Let $X_1,X_2,\ldots,X_k$ with $X_i = \{x_{i1},x_{i2},\ldots\}$ 
be disjoint collections of Boolean variables. Let $f(X_1,\ldots,X_k)$
denote any Boolean formula over (finitely many of) the variables 
from $X_1 \cup\ldots\cup X_k$. Let $Q_k=\exists$ if $k\ge1$ is odd
and $Q_k=\forall$ if $k\ge1$ is even. 
The notation ``$\exists X_i:\ldots$'' means ``there exists an 
assignment of the variables in $X_i$ such that $\ldots$''. 
The analogous remark applies to the notation ``$\forall X_i:\ldots$''. 
Given these notations, we define
\[
\cB_k = \{f(X_1,\ldots,X_k):\
(\exists X_1,\forall X_2,\ldots,Q_k X_k: f(X_1,\ldots,X_k) = 1\} \enspace .
\]
The set consisting of Boolean formulas $f(X_1,\ldots,X_k)$
outside of $\cB_k$ is denoted as $\ol{\cB_k}$.
The subproblem of $\cB_k$ (resp.~of $\ol{\cB_k}$) with $f$ 
being a formula in conjunctive normal form is denoted 
as $\cB_k^{CNF}$ (resp.~as $\ol{\cB_k}^{CNF}$). The 
corresponding subproblems with $f$ being a formula in 
disjunctive normal form are denoted as $\cB_k^{DNF}$ 
and $\ol{\cB_k}^{DNF}$, respectively.
\end{definition}

\begin{theorem}[\cite{S1977}] \label{th:polynomial-hierarchy}
For any $k\ge1$, $\cB_k$ is log-complete in $\Sigma_k^p$.
The same is true for $\cB_k^{CNF}$ if $k$ is odd
and for $\cB_k^{DNF}$ if $k$ is even.
\end{theorem}

\begin{corollary}
For any $k\ge1$, $\ol{\cB_k}$ is log-complete in $\Pi_k^p$. 
This even holds for the set $\ol{\cB_k}^{CNF}$ if $k\ge1$ 
is odd and for the set $\ol{B_k}^{DNF}$ if $k\ge1$ is even.
\end{corollary}

\begin{example} \label{ex:qbf2}
The set $\ol{\cB_2}^{DNF}$, which coincides with the set of 
all Boolean DNF-formulas $f(X_1,X_2)$ satisfying
\[
\forall X_1, \exists X_2: f(X_1,X_2)=0 \enspace .
\]
is $\log$-complete in $\Pi_2^p$.
\end{example}

\subsection{Integer Expressions} \label{subsec:integer-expressions}

\begin{definition}
Let $m\ge1$ be a positive integer. The set $\cE_m$ of
{\em $m$-dimensional unary integer expressions}, simply
called {\em unary integer expressions} if $m$ is clear 
from context, is the smallest set with the following 
properties:
\begin{enumerate}
\item
$\{0,1\}^m \seq \cE_m$. The tuples $(b_1,\ldots,b_m)\in\{0,1\}^m$ 
are called {\em atomic expressions}.
\item
For any $E_1,E_2\in\cE_m$: $(E_1 \cup E_2) , (E_1 + E_2) \in \cE_m$.
\end{enumerate}
\end{definition}
Every expression $E\in\cE_m$ represents a set $L(E) \seq \nats_0^m$
that is defined in the obvious manner.

We briefly note that the classical definition of integer
expressions in~\cite{S1977} is different from ours: there the
expressions define subsets of $\nats_0$, and an atomic expression 
is a binary representation of a single number in $\nats_0$.
In other words, the classical definition deals with 
$1$-dimensional binary expressions whereas we deal with
multi-dimensional unary expressions.

Since ``$\cup$'' is an associative operation, we may simply
write $(E_1 \cup E_2 \cup E_3 \cup\ldots\cup E_s)$ instead
of $(\ldots((E_1 \cup E_2) \cup E_3) \cup\ldots\cup E_s)$.
The analogous remark applies to the operation~``$+$''.

\begin{definition}
An expression $E \in \cE_m$ is said to be
a {\em $(+,\cup)$-expression} if it is a sum of unions 
of atomic expressions. A $(+,\cup)$-expression is called
{\em simple} if every union in the sum is the union of
precisely two (not necessarily different) atomic expressions.
\end{definition}

\begin{example}
The string
\[ E = ((1,1,0)\cup(0,0,0)) + ((1,0,0)\cup(1,0,0)) + ((1,1,1)\cup(0,0,0)) \]
is a simple unary $(+,\cup)$-expression. It represents the set
\[ L(E) = \{(1,0,0),(2,1,0),(2,1,1),(3,2,1)\} \enspace . \]
\end{example}

\subsection{Linear and Semilinear Sets} \label{subsec:linear-sets}

\begin{definition}
The set $L(\vec{c},P) \seq \nats_0^m$ induced by $\vec{c}\in\nats_0^m$ 
and a finite set $P=\{\vec{p_1},\ldots,\vec{p_k}\}\subset\nats_0^m$ is 
defined as 
\[ 
L(\vec{c},P) = \vec{c} + \langle P \rangle\ \mbox{ where }\ 
\langle P \rangle = 
\left\{\sum_{i=1}^{k}a_i\cdot\vec{p_i}: a_i\in\nats_0\right\} 
\enspace .
\]
The elements in $P$ are called {\em periods} and $\vec{c}$ 
is called the {\em constant vector} of $L(\vec{c},P)$.
A subset $L$ of $\nats_0^m$ is called {\em linear} if
$L = L(\vec{c},P)$ for some $\vec{c}\in\nats_0^m$ and 
some finite set $P \subset \nats_0^m$. A {\em semilinear set} 
in $\nats_0^m$ is a finite union of linear sets in $\nats_0^m$.
\end{definition}

\subsection{Containment Problems} \label{subsec:containment-problems}

As mentioned already in the introduction, the {\em containment problem} 
for a family of sets consists in finding an answer to the following 
question: given two sets of the family, 
is the first one a subset of the second one? We will be mainly
concerned with the containment problem for integer expressions and with 
the containment problem for linear and semilinear sets. 
We will assume that the dimension $m$ of sets in $\nats_0^m$ is part
of the input unless we explicitly talk about an $m$-dimensional problem 
for some fixed constant $m$. The following is known:
\begin{enumerate}
\item
The containment problem for $1$-dimensional binary integer
expressions is $\log$-complete in $\Pi_2^p$~\cite{S1977}.
\item
The containment problem for semilinear sets is $\log$-complete
in $\Pi_2^p$~\cite{H1982}. The $\log$-hardness in $\Pi_2^p$
even holds either when numbers are encoded in unary or when 
the dimension is fixed to $1$. 
\item
The containment problem for linear sets is $\log$-complete
in $\Pi_2^p$~\cite{CHH2016}. The $\log$-hardness in $\Pi_2^p$
even holds when numbers are encoded in unary.
\end{enumerate}
The first two hardness results are shown by means of a logspace 
reduction from $\ol{\cB_2}^{DNF}$ to the respective containment 
problem. A suitable modification of Stockmeyer's reduction 
from $\ol{\cB_2}^{DNF}$ to the containment problem for
$1$-di\-men\-sio\-nal binary integer expressions leads to the
following result:

\begin{theorem} \label{th:reduction2intexp}
The containment problem for simple unary 
$(+,\cup)$-expressions is $\log$-complete in $\Pi_2^p$.
\end{theorem}

The proof of this theorem will be given in 
Section~\ref{sec:reduction2intexp}. 

\begin{description}
\item[Notation for Vectors:]
The $j$-th component of a vector $\vec{x}$ is denoted as $x_j$
or, occasionally, as $\vec{x}[j]$. The latter notation is used,
for instance, if there is a sequence of vectors, 
say $\vec{x_1},\ldots,\vec{x_n}$. The $j$-th component of $\vec{x_i}$
is then denoted as $\vec{x_i}[j]$ (as opposed to $x_{i,j}$
or $(x_i)_j$). Throughout the paper, we use $\vec{a^m}$ 
(with $a\in\nats_0$) as a short notation for $(a,\ldots,a) \in \nats_0^m$. 
For instance $\vec{1^m}$ denotes the all-ones vector in $\nats_0^m$. 
The vector with value $1$ in the $i$-th component and zeros 
in the remaining $m-1$ components is denoted as $\vec{e_i^m}$.
The all-ones matrix is denoted as $J$.\footnote{The order of
the matrix will always be clear from context.}
\end{description}

\subsection{Promise Problems} \label{subsec:promise}

A decision problem (without promise) is a problem
with ``yes''- and ``no''-instances. A {\em promise problem}
is a decision problem augmented by a promise that the
input instances passed to an algorithm satisfy a certain
condition. An algorithm needs to solve the promise problem
only on the input instances that satisfy this condition.
It may output anything on the remaining instances.
Hence a promise problem has besides the ``yes''-
and the ``no''-instances a third kind of instances: 
the ones that violate the promised condition.
Decision problem can be viewed as promise problems
with an empty promise. Reductions between promise 
problems should map ``yes''-instances 
(resp.~``no''-instances) of the first problem to 
``yes''-instances (resp.~``no''-instances) of the 
second problem.

\section{Main Results} \label{sec:main-results}

The first result in this section will be concerned with 
the containment problem for linear sets when the latter
is viewed as the following promise problem.
\begin{description}
\item[Instance:] 
dimension $m$, finite sets $P,Q \subset \nats_0^m$, 
vectors $\vec{c},\vec{d}\in\nats_0^m$ and $s\in\{1,\ldots,|P|\}$.
\item[Question:] $L(\vec{c},P) \seq L(\vec{d},Q)$?
\item[Promise:] 
Let $P = \{\vec{p_1},\ldots,\vec{p_k}\}$ and
let $K_s =\left\{\vec{a} \in\{0,1\}^k\left|\ \sum_{i=1}^{k}a_i = s\right.\right\}$. 
With this notation, the following holds:
\begin{equation} \label{eq:promise}
\forall x\in\nats_0^k \sm K_s:
\sum_{i=1}^{k}x_i\cdot\vec{p_i} \in L(\vec{d},Q) \enspace .
\end{equation}
In other words: we make the promise that the 
inclusion $L(\vec{c},P) \seq L(\vec{d},Q)$ can possibly
fail only on linear combinations of $\vec{p_1},\ldots,\vec{p_k}$
with coefficient vectors taken from $K_s$. 
\end{description} 
It is easy to see that we may may set $\vec{d}=\vec{0}$
in this problem without any loss of generality.

In~\cite{CHH2016}, it was shown that the containment problem
for linear sets is $\log$-hard in $\Pi_2^p$. We strengthen
this result by showing that even the corresponding promise
problem exhibits this kind of hardness. This slightly
stronger result will later help us to prove the hardness
of the containment problem for $1$-dimensional linear sets.

\begin{theorem} \label{th:containment-linear}
The containment problem for linear sets is $\log$-hard in $\Pi_2^p$ 
even under the promise~(\ref{eq:promise}) and even when numbers
are encoded in unary. 
\end{theorem}

\begin{proof}
We will describe a logspace reduction from the containment
problem for simple unary $(+,\cup)$-expressions to the 
containment problem for linear sets. An instance of the 
former problem is of the form
\begin{equation} \label{eq:two-simple-expressions}
E = \sum_{i=1}^{s}(\vec{B_{i1}} \cup \vec{B_{i2}})\ \mbox{ and }\
E' = \sum_{i=1}^{s'}(\vec{B'_{i1}} \cup \vec{B'_{i2}})
\end{equation}
where $\vec{B_{i1}},\vec{B_{i2}},\vec{B'_{i1}},\vec{B'_{i2}} \in \{0,1\}^m$. 
Note that we may set $s'=s$ because we could add sum-terms 
of the form $(\vec{0^m}\cup\vec{0^m})$ to the expression which 
has fewer terms. Our goal is to design $(2m+2s)$-dimensional linear 
sets $\vec{c} + \langle P \rangle$ and $\langle P' \cup P'' \rangle$ 
such that 
\begin{equation} \label{eq:reduction2linear} 
L(E) \seq L(E') \dund 
\vec{c} + \langle P \rangle \seq \langle P' \cup P'' \rangle
\enspace.
\end{equation}
Intuitively, we should think of vectors from $\nats_0^{2m+2s}$
as being decomposed into four sections of dimension $m,s,s,m$,
respectively. The first section is called the ``base section'';
the latter three are called ``control sections''. The constant 
vector $\vec{c}$ and the periods 
in $P = \{\vec{p_{ij}}: i\in[s],j\in[2]\}$ are chosen as follows:
\begin{equation} \label{eq:linset1}
\vec{c} = (\vec{0^m},\vec{2^s},\vec{1^s},\vec{1^m})\ \mbox{ and }\ 
\vec{p_{ij}} = (\vec{B_{ij}},\vec{e_i^s},\vec{0^s},\vec{0^m}) \enspace .
\end{equation}
Note that the base section of the periods in $P$ contains 
the atomic sub-expressions of $E$. The vectors 
in $\nats_0^{2m+2s}$ having $(\vec{3^s},\vec{1^s},\vec{1^m})$ 
in their control sections are said to be ``essential''. It is 
evident that
\[ 
L(E)\times\{3\}^s\times\{1\}^s\times\{1\}^m = 
(\vec{c}+\langle P \rangle) \cap (\nats_0^m\times\{3\}^s\times\{1\}^s\times\{1\}^m)
\enspace .
\]
In other words: the set of base sections of the essential vectors
in $\vec{c} + \langle P \rangle$ coincides with $L(E)$. The periods
in $P' = \{\vec{p'_{ij}}: i\in[s],j\in[2]\}$ are similarly defined
as the periods in $P$:
\[ 
\vec{p'_{ij}} = \left\{ \begin{array}{ll}
             (\vec{B'_{ij}},3 \cdot \vec{e_i^s},\vec{e_i^s},\vec{0^m}) & \mbox{if $i\in[s-1]$} \\
             (\vec{B'_{sj}},3 \cdot \vec{e_s^s},\vec{e_i^s},\vec{1^m}) & \mbox{if $i=s$}
           \end{array} \right .  \enspace .
\]
Clearly,
\[
L(E')\times\{3\}^s\times\{1\}^s\times\{1\}^m =
\langle P' \rangle \cap (\nats_0^m\times\{3\}^s\times\{1\}^s\times\{1\}^m)
\enspace .
\]
Note that $L(E) \seq L(E')$ iff any essential vector 
in $\vec{c} + \langle P \rangle$ is contained 
in $\langle P' \rangle$. In order to get the desired
equivalence~(\ref{eq:reduction2linear}), we will design $P''$
such that the following holds:
\begin{description}
\item[Claim 1:]
Any inessential vector from $\vec{c} + \langle P \rangle$ is
contained in $\langle P'' \rangle$.
\item[Claim 2:]
Any essential vector in $\vec{c} + \langle P \rangle$ is contained
in $\langle P' \cup P'' \rangle$ only if it is already
contained in $\langle P' \rangle$.
\end{description}
It is evident that~(\ref{eq:reduction2linear}) is valid
if $P''$ can be defined in accordance with the two above 
claims. Let $n = 1+\max\{x_i: \vec{x} \in L(E),i\in[m]\}$, i.e.,
$n-1$ is the largest number that occurs in a component of 
some vector in $L(E)$. We now set $P'' = P''_1 \cup P''_2$
where
\begin{eqnarray*}
P''_1 & = & 
\{(\vec{0^m},2 \cdot \vec{e_i^s},\vec{1^s},\vec{0^m}) , (\vec{0^m},2 \cdot \vec{e_i^s},\vec{0^s},\vec{0^m}) ,
(\vec{0^m},3 \cdot \vec{e_i^s},\vec{0^s},\vec{0^m}): i\in[s]\} \enspace , \\
P''_2 & = & 
\{(r \cdot \vec{e_i^m},\vec{0^s},\vec{0^s},\vec{e_i^m}) , (n \cdot \vec{e_i^m},\vec{0^s},\vec{0^s},\vec{0^m}):
i\in[m] , r\in\{0,1,\ldots,n-1\}\} \enspace .
\end{eqnarray*}
The proof of the theorem can now be accomplished by showing
that the above two claims are valid for our definition of $P''$
(and by adding some easy observations).
\begin{description}
\item[Proof of Claim 1:]
Let $\vec{x} \in \vec{c} + \langle P \rangle$ be inessential.
An inspection of~(\ref{eq:linset1}) reveals that there must
exist an index $i_0 \in [s]$ such that the $i_0$-th component
of the first control section of $\vec{x}$ has a value that differs
from $3$. Since already the constant vector $\vec{c}$ makes a contribution
of $2$ in this control section, the possible values for $x_{m+i_0}$
are $2,4,5,6,\ldots$. In order to cast $\vec{x}$ as a member 
of $\langle P'' \rangle$, we first pick the 
vector $\vec{u} = (\vec{0^m}, 2 \cdot \vec{e_{i_0}^s},\vec{1^s},\vec{0^m})$. 
Note that $\vec{u}\le\vec{x}$ and $\vec{u}$ already coincides with $\vec{x}$ 
in the second control section. Adding to $\vec{u}$ properly chosen multiples 
of vectors of the form $(\vec{0^m},2 \cdot \vec{e_i^s},\vec{0^s},\vec{0^m})$
or $(\vec{0^m},3 \cdot \vec{e_i^s},\vec{0^s},\vec{0^m})$, we obtain a 
vector $\vec{v} \le \vec{x}$ that coincides with $\vec{x}$ also in the first 
control section. Consider now the entries of $\vec{v}$ and $\vec{x}$ in the 
base section. For any $i\in[m]$, consider the 
decomposition $x_i-v_i = q_i n + r_i$ with $q_i\ge0$
and $0 \le r_i \le n-1$. Adding to $v$ the vector
\[ 
\sum_{i=1}^{m}\left(q_i \cdot (n \cdot \vec{e_i^m},\vec{0^s},\vec{0^s},\vec{0^m}) + 
(r_i \cdot \vec{e_i^m},\vec{0^s},\vec{0^s},\vec{e_i^m})\right) \enspace ,
\]
we obtain a vector that coincides with $\vec{x}$ (since, by now, 
it also coincides with $\vec{x}$ in the base section and in the 
third control section). 
\item[Proof of Claim 2:]
Let $\vec{x} \in \vec{c} + \langle P \rangle$ be essential and suppose
that $\vec{x} \in \langle P' \cup P'' \rangle$. A representation
of $\vec{x}$ as a member of $\langle P' \cup P'' \rangle$ cannot 
make use of a vector of the 
form $(\vec{0^m},2 \cdot \vec{e_i^s},\vec{1^s},\vec{0^m})$ 
because there is no way to extend the value $2$ in the $i$-th 
component of the first control section to $3$ (since any period 
in $P' \cup P''$ adds either $0$ or a value greater than $1$ 
to this component). Given that we do not employ these vectors, 
it follows that any representation of $\vec{x}$ as a member 
of $\langle P' \cup P'' \rangle$ must be of the 
form $\vec{x} = \vec{x'} + \vec{x''}$ for some essential 
vector $\vec{x'} \in \langle P' \rangle$ and some 
vector $\vec{x''} \in \langle P' \cup P'' \rangle$ (because,
without employing an essential vector from $\langle P' \rangle$,
we wouldn't get $\vec{1^s}$ into the second control section).
Since $\vec{x'}$ is essential, it will already 
contribute $(\vec{3^s},\vec{1^s},\vec{1^m})$ to the three control sections. 
It follows that $\vec{x''}=\vec{0^{2m+2s}}$ because adding any period 
from $P' \cup P''$ to $\vec{x'}$ will destroy the 
pattern $(\vec{3^s},\vec{1^s},\vec{1^m})$ in the control sections 
or will induce a component of value at least $n$ in the base section 
(which is larger than any component of $\vec{x}$ in the base section).
It follows that $\vec{x}=\vec{x'} \in \langle P' \rangle$.
\end{description}
It can be shown by standard arguments that the 
transformation $(E,E') \mapsto (\vec{c},P,P',P'')$
is logspace-computable (even when numbers are encoded
in unary). Finally observe that the above definition 
of essential vectors implies that every essential 
vector from $\vec{c}+\langle P \rangle$ employs
a coefficient vector from $\{0,1\}^{|P|}$ with
precisely $s$ ones. Since any inessential vector
from $\vec{c}+\langle P \rangle$ also belongs 
to $\langle P' \rangle \seq \langle P' \cup P'' \rangle$, 
the promised condition~(\ref{eq:promise}) is satisfied
(with $P' \cup P''$ at the place of $Q$).
This concludes the proof.
\end{proof}

We will show in the sequel that the containment problem
for $1$-dimensional linear sets (with numerical input
parameters given in binary representation) is $\log$-hard 
in $\Pi_2^p$. To this end, we will make use of the following
result on the aggregation of diophantine equations:

\begin{lemma}[\cite{GW1972}] \label{lem:aggregation}
Let 
\begin{equation} \label{eq:2equations} 
\sum_{j=1}^{r}a_{1j}x_j = b_1\ \mbox{ and }\ 
\sum_{j=1}^{r}a_{2j}x_j = b_2 
\end{equation}
be a system of two linear diophantine equations where 
$a_{1j},a_{2j}$ are non-negative integers and $b_1,b_2$ 
are strictly positive integers. Let $t_1,t_2$
be positive integers satisfying the following conditions:
\begin{enumerate}
\item
$t_1$ and $t_2$ are relatively prime.
\item
$t_1$ does not divide $b_2$ and $t_2$ does not divide $b_1$.
\item
$t_1 > b_2-a_2$ and $t_2 > b_1-a_1$ where $a_i$ denotes the
smallest nonzero coefficient in $\{a_{i1},\ldots,a_{ir}\}$.
\end{enumerate}
Then, restricting $x_j$ to non-negative integers, the solution
set of~(\ref{eq:2equations}) is the same as the solution set of
\[ 
t_1\cdot\sum_{j=1}^{r}a_{1j}x_j + t_2\cdot\sum_{j=1}^{r}a_{2j}x_j
= t_1 \cdot b_1 + t_2 \cdot b_2 \enspace .
\]
\end{lemma}

\noindent
Note that 
\begin{equation} \label{eq:aggregation-coefficients}
t_1 = 1+\max\{b_1,b_2\}\ \mbox{ and } t_2 = 1+t_1 
\end{equation}
is among the choices for $t_1,t_2$ such that the three
conditions mentioned in Theorem~\ref{lem:aggregation}
are satisfied. An iterative application of 
Lemma~\ref{lem:aggregation} leads to the following result:

\begin{corollary} \label{cor:aggregation}
Let $A\vec{x}=\vec{b}$ with $A\in\nats_0^{m \times r}$
and $b\in\nats^m$ be a system of linear diophantine
equations. Let $A_1,\ldots,A_m\in\nats_0^r$ denote 
the rows of $A$. Then there exist ``aggregation 
coefficients'' $t_1,\ldots,t_m\in\nats$ such that
the solution set for $A\vec{x}=\vec{b}$ within $\nats_0^r$
is the same as the solution set for the single equation
\[ 
\left(\sum_{i=1}^{m}t_iA_i\right)\vec{x} = \sum_{i=1}^{m}t_ib_i
\enspace .
\]
\end{corollary}

\noindent
We are ready now for the next result:

\begin{theorem} \label{th:hardness-dim1-linear-sets}
The containment problem for $1$-dimensional linear sets
is $\log$-hard in $\Pi_2^p$.
\end{theorem}

\begin{proof}
We will describe a logspace reduction from the containment 
problem for multidimensional linear sets (under the
promise~(\ref{eq:promise})) to the containment problem 
for $1$-dimensional linear sets.
The proof proceeds in stages. In Stage 1, we express 
the containment problem for multidimensional linear sets
in terms of a system of diophantine equations. The latter
arises naturally when the the containment problem is written
in matrix notation. In Stage~2, we rewrite the system of
diophantine equations so that Corollary~\ref{cor:aggregation}
comes into play and an aggregation into a single equation
takes place. Although the solution set of the single equation 
and the one of the original system of equations do not fully 
coincide, they are related sufficiently closely so that, in 
Stage 3, we finally obtain the desired logspace reduction. \\
Let us start with Stage~1. Let $\vec{c}\in\nats_0^m$,
$P = \{\vec{p_1},\ldots,\vec{p_k}\} \seq \nats_0^m$
and $Q = \{\vec{q_1},\ldots,\vec{q_\ell}\} \seq \nats_0^m$
form an instance of the containment problem. We set
\[
A := [\vec{p_1}\ \ldots\ \vec{p_k}] \in \nats_0^{m \times k}
\mbox{ and }
B := [\vec{q_1}\ \ldots\ \vec{q_\ell}] \in \nats_0^{m \times \ell}
\enspace ,
\]
which yields two matrices with non-negative entries.
With this notation:
\[
L(\vec{c},P) \seq L(\vec{0},Q) \dund
\forall \vec{x}\in\nats_0^k, \exists \vec{y}\in\nats_0^\ell:
\vec{c}+A\vec{x} = B\vec{y} \enspace .
\]
In matrix notation, the promise~(\ref{eq:promise}) 
can be written as follows:
\begin{equation} \label{eq2:promise}
\forall \vec{x} \in \nats_0^k \setminus K_s,
\exists \vec{y}\in\nats_0^\ell: \vec{c}+A\vec{x} = B\vec{y}
\enspace .
\end{equation}
We now proceed to Stage~2 of the proof. We say that two systems 
of diophantine equations with the same collection $(\vec{x},\vec{y})$
of variables are {\em fully equivalent} if their solution sets 
in $\nats_0^{k}\times\nats_0^\ell$ are the same.
We say that they are {\em $s$-equivalent} if their
solution sets in $K_s\times\nats_0^\ell$ are the same.
We use the symbol ``$\equiv$'' for ``full equivalence''
and the symbol ``$\sequiv$'' for ``$s$-equivalence''.
$M_i$ denotes the $i$-th row of a matrix $M$
and $u$ denotes the largest entry of the matrix $A$.
With this notation, we may rewrite the equation
system $\vec{c}+A\vec{x} = B\vec{y}$ as follows:
\begin{eqnarray*}
\vec{c}+A\vec{x} = B\vec{y} & \equiv &
[-A\ \ B]\left(\begin{array}{c}x \\ y\end{array}\right) = \vec{c} \\
& \equiv &
[(uJ-A)\ \ B]\left(\begin{array}{c}x \\ y\end{array}\right) =
\vec{c} + uJ\cdot\vec{x} \\
& \sequiv &
[(uJ-A)\ \ B]\left(\begin{array}{c}x \\ y\end{array}\right) =
\vec{c} + su\cdot\vec{1^m} \\
& \stackrel{Cor.~\ref{cor:aggregation}}{\equiv} &
\left(\sum_{i=1}^{m}t_i(uJ-A)_i\right)\cdot\vec{x} +
\left(\sum_{i=1}^{m}t_iB_i\right)\cdot\vec{y} \\
&& \mbox{} = \sum_{i=1}^{m}t_ic_i + su\left(\sum_{i=1}^{m}t_i\right) \\
& \sequiv &
\sum_{i=1}^{m}t_ic_i +
\left(\sum_{i=1}^{m}t_iA_i\right)\cdot\vec{x} =
\left(\sum_{i=1}^{m}t_iB_i\right)\cdot\vec{y}
\end{eqnarray*}
We move on to Stage~3 of the proof where we obtain the 
desired reduction to the 1-dimensional containment problem:
\begin{eqnarray*}
L(\vec{c},P) \seq L(\vec{0},Q) & \dund &
\forall\vec{x}\in\nats_0^{k}, \exists\vec{y}\in\nats_0^\ell:
\vec{c}+A\vec{x} = B\vec{y} \\
& \dund &
\forall\vec{x}\in\nats_0^{k}, \exists\vec{y}\in\nats_0^\ell: \\
&& \mbox{}
\underbrace{\sum_{i=1}^{m}t_ic_i}_{=:c} +
\underbrace{\left(\sum_{i=1}^{m}t_iA_i\right)}_{=:(p_1,\ldots,p_k)=:\vec{p}^\top}
\cdot\vec{x} =
\underbrace{\left(\sum_{i=1}^{m}t_iB_i\right)}_{=:(q_1,\ldots,q_\ell)=:\vec{q}^\top}\cdot\vec{y} \\
& \dund &
L(c,\vec{p}) \seq L(0,\vec{q})
\end{eqnarray*}
The validity of the 2nd equivalence follows by case analysis:
\begin{description}
\item[Case 1:] $\vec{x} \in K_s$. \\
We make use of the relation $\sequiv$ that we had established
before: within the restricted domain $K_s\times\nats_0^\ell$, 
the solution set for the equation system $\vec{c}+A\vec{x} = B\vec{y}$ 
coincides with solution set for the single equation 
$c+\vec{p}^\top\vec{x} = \vec{q}^\top\vec{y}$.
\item[Case 2:] $\vec{x} \in \nats^k \sm K_s$. \\
We make use of the promise~(\ref{eq2:promise})
and of the observation that the validity
of $\vec{c}+A\vec{x}=B\vec{y}$ implies the validity of 
$c+\vec{p}^\top\vec{x} = \vec{q}^\top\vec{y}$ (because
the aggregation of valid equations always yields another
valid equation).
\end{description}
It is easy to see that the reduction mapping
$(\vec{c},P,Q) \mapsto (c,\vec{p},\vec{q})$ 
is logspace-computable. This concludes the proof. 
\end{proof}

Combining the restrictions of dimensionality $1$ and unary
encoding of numbers, the containment problem for linear sets 
becomes solvable in polynomial time:

\begin{theorem}
The containment problem for $1$-dimensional linear sets with
a unary encoding of numbers is in $P$.
\end{theorem}

\begin{proof}
Consider an input instance given by (the unary encoding 
of) $c,P,c',P'$ with $c,c' \in \nats_0$ and $P,P' \subset \nats$.
Let $g$ (resp.~$g'$) be the greatest common divisor of the
periods in $P$ (resp.~in $P'$). We make the following observation:
\begin{description}
\item[Claim:]
The containment $c + \langle P \rangle \seq c' + \langle P' \rangle$ 
is possible only if $c' \le c$ and if $g'$ is a divisor of $g$ and 
of $c-c'$. 
\end{description}
Given the assertion in the claim, we can accomplish the proof
as follows. Setting $c_0 = c-c'$, our original question, 
``$c + \langle P \rangle \seq c' + \langle P' \rangle$?'', 
is equivalent to ``$c_0 + \langle P \rangle \seq \langle P' \rangle$?''.
We may now even assume that $g' = 1$ (because, if necessary,
we can divide all numerical parameters by $g'$). If $1$ is 
among the periods of $P'$, then the answer
to ``$c_0 + \langle P \rangle \seq \langle P' \rangle$?''
is clearly ``yes''. Suppose now that $1 \notin P'$.
It is well known that $\langle P' \rangle$ contains all 
but finitely many natural numbers~\cite{RG-S2009}. 
Let $F(P')$ (called the {\em Frobenius number} of $P'$) 
denote the largest number in $\nats$ that is not 
contained in $\langle P' \rangle$. It is well known 
that $F(P') < (\max(P')-1)\cdot(\min(P')-1)$~\cite{B1942}. 
The questions ``$x \in c_0 + \langle P \rangle$?''
and ``$x \in \langle P' \rangle$?'' can be answered
for all $x < (\max(P')-1)\cdot(\min(P')-1)$ in the obvious 
way by dynamic programming. Given the answers to these questions,
we can immediately decide 
whether $c_0 + \langle P \rangle \seq \langle P' \rangle$. \\
All that remains to be done is proving the above claim.
Suppose that 
\begin{equation} \label{eq:ass-containment}
c + \langle P \rangle \seq c' + \langle P' \rangle \enspace .
\end{equation}
This obviously implies that $c' \le c$. It is furthermore obvious
that $\langle P \rangle \seq g\cdot\nats_0$ 
and $\langle P' \rangle \seq g'\cdot\nats_0$. Moreover, by the 
definition of the Frobenius number, 
$s := g \cdot F\left((\frac{1}{g} \cdot \langle P \rangle\rangle\right)$
is the largest multiple of $g$ that does not belong 
to $\langle P \rangle$. Hence $c+s+g,c+s+2g \in c+\langle P \rangle$
and, because of~(\ref{eq:ass-containment}), there must 
exist $q_2 > q_1 \ge 1$ such that $c+s+g = c' + q_1 g'$
and $c+s+2g = c' + q_2 g'$. Now we obtain $g = (q_2-q_1)g'$
so that $g'$ is a divisor of $g$. Since
$(c-c') + \langle P \rangle \seq \langle P' \rangle \seq g'\cdot\nats_0$
and $\langle P \rangle$ contains only multiples of $g'$
(because it only contains multiples of $g$), it follows 
that $g'$ must also be a divisor of $c-c'$, which concludes 
the proof of the claim and the proof of the theorem.
\end{proof}

\section{Proof of Theorem~\ref{th:reduction2intexp}}
\label{sec:reduction2intexp}

It is easy to see that the containment problem for simple unary 
$(+,\cup)$-expressions is a member of the complexity class $\Pi_2^p$. 
In somewhat more detail, let $E$ and $E'$ be two simple unary
expressions of the form~(\ref{eq:two-simple-expressions}). 
Then $L(E) \seq L(E')$ iff
\[ 
\forall a \in\{1,2\}^s, \exists a'\in\{1,2\}^{s'}:
\sum_{i=1}^{s}\vec{B_{ia_i}} = \sum_{i=1}^{s'}\vec{B'_{ia'_i}}
\enspace .
\]
The membership in $\Pi_2^p$ is now immediate from a well known
characterization of $\Pi_2^p$ due to Wrathall~\cite{W1976}: 
$L \in \Pi_2^p$ iff there exists a polynomial $q$ and a 
language $L_0\in\mbox{P}$ such that
\[ L = \{x| 
(\forall y_1\mbox{ with }|y_1| \le q(|x|))
(\exists y_2\mbox{ with }|y_2| \le q(|x|)): 
\langle y_1,y_2,x \rangle \in L_0\} \enspace .
\]

It remains to show that it is $\log$-hard in $\Pi_2^p$. To this end, 
we will design 
a logspace reduction from $\ol{\cB_2}^{DNF}$ to this problem. 
Let $f(X_1,X_2)$ be an instance of $\ol{\cB_2}^{DNF}$
(as described in Example~\ref{ex:qbf2}). Since $f$ employs 
only finitely many variables, we may assume 
that $X_i = \{x_{i1},\ldots,x_{in}\}$ for $i=1,2$ and some $n\ge1$. 
As a DNF-formula, $f$ is the disjunction of Boolean monomials,
say $f = M_1 \vee\ldots\vee M_m$. We may clearly assume that
none of the monomials contains the same variable twice. We will 
transform $f(X_1,X_2)$ into simple unary $(+,\cup)$-expressions $E_1$
and $E_2$ such that
\begin{equation} \label{eq:b2-reduction}
(\forall X_1, \exists X_2: f(X_1,X_2)=0) \dund (L(E_1) \seq L(E_2))
\enspace .
\end{equation}
For all $i=1,\ldots,n$ and $j=1,\ldots,m$, let
\[ 
\vec{b_{1i}}[j] = \left\{ \begin{array}{ll}
          1 & \mbox{if $x_{1i} \in M_j$} \\
          0 & \mbox{otherwise}
         \end{array} \right. \enspace ,
\]
i.e., the binary vector $\vec{b_{1i}}\in\{0,1\}^m$ indicates in which monomials
the variable $x_{1i}$ actually occurs. Let $\vec{b'_{1i}}\in\{0,1\}^m$
denote the corresponding vector with indicator bits for the
occurrences of $\ol{x_{1i}}$ within $M_1,\ldots,M_m$.
Let the vectors $\ol{\vec{b_{1i}}}$ and $\ol{\vec{b'_{1i}}}$ be 
obtained from $\vec{b_{1i}}$ 
and $\vec{b'_{1i}}$, respectively, by bitwise negation. Clearly, the 
bits of these vectors indicate the non-occurrences of $x_{1i}$
resp.~$\ol{x_{1i}}$ within $M_1,\ldots,M_m$. 
Let $\vec{b_{2i}}, \vec{b'_{2i}}, \ol{\vec{b_{2i}}}, \ol{\vec{b'_{2i}}}$ be the
corresponding vectors with indicator bits for the occurrences
resp.~non-occurrences of the variable $x_{2i}$. We now define 
a couple of $(+,\cup)$-expressions:
\begin{eqnarray*}
E'_1 = \sum_{i=1}^{n}(\vec{1^m}\cup\vec{1^m}) & \mbox{and} &  
E_1 = E'_1 + \sum_{i=1}^{n}(\ol{\vec{b_{1i}}} \cup \ol{\vec{b'_{1i}}}) \\
E'_2 = \sum_{j=1}^{m}\sum_{i=1}^{2n-1}(\vec{e_j^m} \cup \vec{0^m}) & \mbox{and} &
E_2  = E'_2 + \sum_{i=1}^{n}(\vec{b_{2i}} \cup \vec{b'_{2i}}) \enspace .
\end{eqnarray*}
The following immediate observations will prove useful:
\begin{enumerate}
\item
$L(E'_1) = \{n\cdot\vec{1^m}\}$ and $L(E'_2) = \{0,\ldots,2n-1\}^m$. 
\item
$L(E_1) \seq \{n,\ldots,2n\}^m$ and $L(E_2) \supeq \{n,\ldots,2n-1\}^m$. 
\end{enumerate}
Note that the only vectors of $L(E_1)$ which might perhaps not belong 
to $L(E_2)$ are the ones with at least one component of size $2n$.
The following definitions take care of these ``critical vectors''.
We say that a partial assignment of the variables in $X_1 \cup X_2$
{\em annuls} $M_j$ if one of the literals contained in $M_j$
is set to $0$. Let $\vec{y}\in\{n,\ldots,2n\}^m$. 
An assignment $A_1:X_1\ra\{0,1\}$ is said to be an
{\em $X_1$-assignment of type $\vec{y}$} if the following holds: 
\[ 
\forall j=1,\ldots,m: (\vec{y}[j] = 2n \dund A_1\mbox{ does not annul }M_j) 
\enspace .
\] 
We say that $A_2:X_2\ra\{0,1\}$ is an {\em $X_2$-assignment of type $\vec{y}$}
if the following holds:
\[ 
\forall j=1,\ldots,m: (\vec{y}[j] = 2n \impl A_2\mbox{ annuls }M_j)
\enspace .
\]
The desired equivalence~(\ref{eq:b2-reduction}) is easy to derive
from the following claims:
\begin{description}
\item[Claim 1:]
For every $\vec{y} \in L(E_1)$, there exists an $X_1$-assignment $A_1$
of type $\vec{y}$.
\item[Claim 2:]
For every $A_1:X_1\ra\{0,1\}$, there exists $\vec{y} \in L(E_1)$
such that $A_1$ is an $X_1$-assignment of type $\vec{y}$.
\item[Claim 3:] 
For every $\vec{y} \in \{n,\ldots,2n\}^m$:
\[
\vec{y} \in L(E_2) \dund (\exists A_2:X_2\ra\{0,1\}:
\mbox{$A_2$ is an $X_2$-assignment of type $\vec{y}$}) \enspace .
\]
\item[Proof of Claim 1:]
Pick any $\vec{y} \in L(E_1)$. It follows that $\vec{y}$ is of the form 
\begin{equation} \label{eq:e1}
\vec{y} = n\cdot\vec{1^m} + \sum_{i=1}^{n}\widetilde{\vec{b_{1i}}}\ 
\mbox{ with }\ 
\widetilde{\vec{b_{1i}}} \in \{\ol{\vec{b_{1i}}},\ol{\vec{b'_{1i}}}\}
\enspace .
\end{equation}
If $\widetilde{\vec{b_{1i}}} = \ol{\vec{b_{1i}}}$,
we set $A_1(x_{1i})=0$ else, if $\widetilde{\vec{b_{1i}}} = \ol{\vec{b'_{1i}}}$,
we set $A_1(x_{1i})=1$. We claim that $A_1$ is of type $\vec{y}$.
This can be seen as follows. Pick any $j\in\{1,\ldots,m\}$. 
An inspection of~(\ref{eq:e1}) reveals the following:
\begin{itemize}
\item
Suppose that $\vec{y}[j]=2n$. It follows that $\widetilde{\vec{b_{1i}}}[j] = 1$ 
for $i=1,\ldots,n$. Hence, if $\widetilde{\vec{b_{1i}}} = \ol{\vec{b_{1i}}}$, 
then $A_1(x_{1i})=0$, $\ol{\vec{b_{1i}}}[j]=1$ and, 
therefore, $x_{1i} \notin M_j$. Similarly, 
if $\widetilde{\vec{b_{1i}}} = \ol{\vec{b'_{1i}}}$, 
then $A_1(x_{1i})=1$, $\ol{\vec{b'_{1i}}}[j]=1$ and, 
therefore, $\ol{x_{1i}} \notin M_j$. 
Since these observations hold for all $i=1,\ldots,n$, 
we may conclude that $A_1$ does not annul $M_j$.
\item
Suppose that $\vec{y}[j] \le 2n-1$. Then there exists $i\in\{1,\ldots,n\}$ 
such that $\widetilde{\vec{b_{1i}}}[j] = 0$.
Hence, if $\widetilde{\vec{b_{1i}}} = \ol{\vec{b_{1i}}}$,
then $A_1(x_{1i})=0$, $\ol{\vec{b_{1i}}}[j]=0$ and,
therefore, $x_{1i} \in M_j$. Similarly,
if $\widetilde{\vec{b_{1i}}} = \ol{\vec{b'_{1i}}}$,
then $A_1(x_{1i})=1$, $\ol{\vec{b'_{1i}}}[j]=0$ and,
therefore, $\ol{x_{1i}} \in M_j$.
It follows that $A_1$ does annul $M_j$.
\end{itemize}
The above discussion shows that $A_1$ is of type $\vec{y}$, indeed.
\item[Proof of Claim 2:]
Given any $A_1:X_1\ra\{0,1\}$, we 
set $\vec{y} = n\cdot\vec{1^m} + \sum_{i=1}^{n}\widetilde{\vec{b_{1i}}}$
where $\widetilde{\vec{b_{1i}}} = \ol{\vec{b_{1i}}}$ if $A_1(x_{1i})=0$ and, similarly, 
$\widetilde{\vec{b_{1i}}} = \ol{\vec{b'_{1i}}}$ if $A_1(x_{1i})=1$. Note that, with this
definition of $\vec{y}$, $A_1$ is precisely the $X_1$-assignment that 
we had chosen in the proof of Claim~1. As argued in the proof of
Claim~1 already, $A_1$ is of type $\vec{y}$.
\item[Proof of Claim 3:]
Pick any $\vec{y}\in\{n,\ldots,2n\}^m$. Suppose first 
that $\vec{y} \in L(E_2)$. It follows that $\vec{y}$ 
is of the form
\begin{equation} \label{eq:e2}
\vec{y} = \vec{y'} + \sum_{i=1}^{n}\widetilde{\vec{b_{2i}}}\ 
\mbox{ with }\ \vec{y'} \in \{0,\ldots,2n-1\}^m\ \mbox{ and }\  
\widetilde{\vec{b_{2i}}} \in \{\vec{b_{2i}},\vec{b'_{2i}}\}
\enspace .
\end{equation} 
If $\widetilde{\vec{b_{2i}}} = \vec{b_{2i}}$, we set $A_2(x_{2i})=0$ 
else, if $\widetilde{\vec{b_{2i}}} = \vec{b'_{2i}}$,
we set $A_2(x_{2i})=1$. We claim that $A_2$ is of type $\vec{y}$. 
Consider an index $j\in\{1,\ldots,m\}$ such that $\vec{y}[j]=2n$.
An inspection of~(\ref{eq:e2}) reveals that there 
exists $i\in\{1,\ldots,n\}$ such that $\widetilde{\vec{b_{2i}}}[j] = 1$.
If $\widetilde{\vec{b_{2i}}} = \vec{b_{2i}}$,
then $A_2(x_{2i})=0$, $\vec{b_{2i}}[j]=1$ and, 
therefore, $x_{2i} \in M_j$. Similarly, 
if $\widetilde{\vec{b_{2i}}} = \vec{b'_{2i}}$,
then $A_2(x_{2i})=1$, $\vec{b'_{2i}}[j]=1$ and, 
therefore, $\ol{x_{2i}} \in M_j$. In any case, $A_2$ annuls $M_j$ 
and we may conclude that $A_2$ is of type $\vec{y}$. \\
Suppose now that there exists an $X_2$-assignment $A_2$ that is 
of type $\vec{y}\in\{n,\ldots,2n\}^m$. 
We define $\vec{y''} = \sum_{i=1}^{n}\widetilde{\vec{b_{2i}}}$ 
where $\widetilde{\vec{b_{2i}}} = \vec{b_{2i}}$ if $A_2(x_{2i})=0$ and, similarly,
$\widetilde{\vec{b_{2i}}} = \vec{b'_{2i}}$ if $A_2(x_{2i})=1$. Since $A_2$ is of
type $\vec{y}$, it annuls every $M_j$ with $\vec{y}[j]=2n$. It follows that,
for every $j\in\{1,\ldots,m\}$ with $\vec{y}[j]=2n$, there 
exists $i\in\{1,\ldots,n\}$ such either $x_{2i} \in M_j$ 
and $A_2(x_{2i})=0$ or $\ol{x_{2i}} \in M_j$ and $A_2(x_{2i})=1$. 
In both cases, we have that $\widetilde{\vec{b_{2i}}}[j] = 1$. It follows
from this discussion that $\vec{y''}[j]\ge1$ for every $j$ with $\vec{y}[j]=2n$.
Obviously $\vec{y''}[j] \le n$ for all $j=1,\ldots,m$. 
Since $L(E'_2) = \{0,\ldots,2n-1\}^m$ and $\vec{y}\in\{n,\ldots,2n\}^m$, 
there exists $\vec{y'} \in L(E_2)$ such that $\vec{y} = \vec{y'}+\vec{y''}$.
This decomposition of $\vec{y}$ shows that $\vec{y} \in L(E_2)$.
\end{description}
We are ready now for proving~(\ref{eq:b2-reduction}). Assume first
that the condition on the left hand-side of~(\ref{eq:b2-reduction}) 
is valid. Pick any $\vec{y} \in L(E_1)$. Pick an $X_1$-assignment $A_1$
of type $\vec{y}$ (application of Claim~1). It follows that the monomials $M_j$ 
with $\vec{y}[j]=2n$ are not yet annulled by $A_1$. According to the left hand-side 
of~(\ref{eq:b2-reduction}), there must exist an assignment $A_2:X_2\ra\{0,1\}$ 
that annuls them. In other words: $A_2$ is an $X_2$-assignment of type $\vec{y}$.
We may now conclude from Claim~3 that $\vec{y} \in L(E_2)$, as desired. \\
Suppose now that $L(E_1) \seq L(E_2)$. Pick any assignment $A_1:X_1\ra\{0,1\}$. 
Pick $\vec{y} \in L(E_1)$ such $A_1$ is an $X_1$-assignment of type $\vec{y}$ (application 
of Claim~2). It follows that only the monomials $M_j$ with $\vec{y}[j]=2n$ are not 
yet annulled by $A_1$. Since $\vec{y}$, as an element of $L(E_1)$, must
satisfy $\vec{y}\in\{n,\ldots,2n\}^m$ and must furthermore belong to $L(E_2)$, 
we may conclude from Claim~3 that there exists an $X_2$-assignment
$A_2:X_2\ra\{0,1\}$ of type $\vec{y}$. In other words: $A_2$ annuls all
monomials $M_j$ with $\vec{y}[j]=2n$. It follows from this discussion
that the condition on the left hand-side of~(\ref{eq:b2-reduction}) 
is valid, which concludes the proof.

\section{Open Problems} \label{sec:open-problems}

In the proof of our hardness results, we made essential use 
of the fact that $\langle P \rangle$ contains all linear combinations
of the periods in $P$ with coefficient vectors from $\nats_0^{|P|}$.
We would be interested to know whether the computational
complexity of the containment problem is still the same
when we deal with coefficient vectors from $\nats^{|P|}$
(thereby ruling out $0$-coefficients).

\subparagraph*{Acknowledgements.}

Many thanks go to Dmitry Chistikov and Christoph Haase who pointed 
my attention to~\cite{CHH2016}, a paper that (without mentioning 
this explicitly) yields the $\log$-hardness of the containment 
problem for linear sets of variable dimension.





\end{document}